\documentclass[10pt,conference]{IEEEtran}
\usepackage{graphicx,cite,amssymb,amsmath,xcolor,subfig}
\IEEEoverridecommandlockouts      
\newtheorem{thm}{Theorem}
\newtheorem{prop}{Proposition}
\newtheorem{corr}{Corollary}
\newtheorem{fact}{Fact}

\newcommand{\be}{\begin{equation}}
\newcommand{\ee}{\end{equation}}
\newcommand{\ben}{\begin{equation*}}
\newcommand{\een}{\end{equation*}}
\newcommand{\mc}{\mathcal}

\newcommand{\kap}{\kappa}
\newcommand{\mbf}[1]{\mathbf{#1}}

\newcommand{\expec}{\mathbb{E}}
\begin{document}
\title{Rewritable Storage Channels with Hidden State}
\author{
\authorblockN{R. Venkataramanan}
\authorblockA{
University of Cambridge, UK \\
Email: ramji.v@eng.cam.ac.uk \vspace{-8pt}}
\and
\authorblockN{S. Tatikonda}
\authorblockA{
Yale University, USA \\
Email: sekhar.tatikonda@yale.edu \vspace{-8pt}}
\and
\authorblockN{L. A. Lastras-Monta{\~n}o, M. Franceschini}
\authorblockA{IBM T. J. Watson Research Center, USA\\
Email: \{lastrasl,franceschini\}@us.ibm.com \vspace{-8pt}
}
\thanks{This work was supported in part by NSF grant CCF-1017744 and an IBM faculty award.
A part of the paper appeared in the proceedings of the 2012 IEEE International Symposium on Information Theory \cite{RVRew12}. }
}

\maketitle
\begin{abstract}
Many storage channels  admit reading and rewriting of the content at a given cost. We consider rewritable channels with a hidden state which models the unknown characteristics of the memory cell.  In addition to mitigating the effect of the write noise, rewrites can help the write controller  obtain a better estimate of the hidden state. The paper has two contributions. The first is  a lower bound on the capacity of a general rewritable channel with hidden state. The lower bound is obtained using a coding scheme that combines Gelfand-Pinsker coding with superposition coding. The rewritable AWGN channel is discussed as an example. The second contribution is a simple coding scheme for  a rewritable channel where the write noise and hidden state are both uniformly distributed. It is shown that this scheme is asymptotically  optimal as the number of rewrites gets large.\end{abstract}

\section{Introduction} \label{sec:intro}
 In  nonvolatile memory technologies, the write mechanism is commonly impaired by {write noise} due to which the  value written on a cell is  different from the one intended. An important feature of many of these technologies such as Flash \cite{Flash03}, Phase Change Memory \cite{PCM10} and Resistive RAM \cite{Memris08,RRam07} is that they allow rewriting, i.e. the value written on a  memory cell can be read and rewritten if necessary. Rewrites can increase the storage capacity but are costly since they are time consuming and degrade the memory. Hence there is a fundamental trade-off between the number of writes and the amount of information that can be stored in a memory cell.

 Given a memory array of $n$ cells, the goal is to maximize the number of distinct messages that can be reliably encoded in the array, subject to a constraint on the average or maximum number of writes per cell. The cells are assumed to be statistically independent.
Rewritable channels were introduced in \cite{Rew08} and subsequently studied in \cite{DatDepN,RewSuperpos,UnifCap10,RewJour12} under an average cost constraint. Maximum cost constrained rewritable channels were considered in \cite{RewMaxCost,BunteLap10,BunteLap11,WeissAction}.

In practice, a memory cell is an amalgam of physical components which reacts to inputs in some way that designers hope to model as well as possible. However, there are always some unknown characteristics of the cell due, for example, to fabrication variability. These characteristics, which may be too costly to learn, introduce an extra degree of uncertainty into the value written on the cell.  In this paper, we model this effect with the channel $P_{Y|X,S}$ where $X$ is the input stimulus, $S$ is a hidden (unknown) state parameter of the cell and $Y$ is the value stored in the cell. $S$ is assumed to be have known distribution $P_S$. The alphabets of $X,Y,S$ are denoted by $\mc{X},\mc{Y}, \mc{S}$, respectively.
We consider two canonical examples of rewritable channels  in this paper:
\begin{enumerate}
\item \emph{Uniform noise channel with state}: The channel model is
\be
Y= X+ W + S.
\label{eq:unif_state_model}
\ee
$X \in [0,1]$ is the input stimulus, the write noise $W$ is uniformly distributed in $[-a/2,a/2]$,  and the state $S$ is uniformly distributed in $[0,B]$.
$a,B$ are known positive constants. The basic version of this channel ($S=0$) was introduced in  \cite{Rew08} as a simple  model that captures some essential features of non-volatile memories such as analog inputs and bounded write noise.

\item \emph{AWGN channel with state}:
The channel is again described by the additive model \eqref{eq:unif_state_model}. The write noise and the state are Gaussian random variables: $W$ is distributed as
$\mc{N}(0,N)$ and  $S$ is distributed as $\mc{N}(0, \sigma^2_s)$.\footnote{$\mc{N}(\mu,\sigma^2)$ is the Gaussian distribution with mean $\mu$ and variance $\sigma^2$.} The input is constrained in terms of the  either the average or peak power per write.

\end{enumerate}

A key feature of the state $S$ is that it stays \emph{fixed} across write attempts in each cell. Conditioned on $S=s$, the value stored in the cell after each attempt depends only on the most recent input stimulus --  it is determined according to $P_{Y|XS}(.|X,S=s)$, where $X$ is the current input. In the additive model \eqref{eq:unif_state_model}, this means that each write attempt on a cell is affected with an independent realization of the random variable $W$, while $S$  stays {fixed} across write attempts.

In this paper, we consider rewritable channels with a constraint on the the average number of writes per cell. Given a constraint $\kappa$ on the average write cost, the goal is to determine the capacity $C(\kappa)$ and design coding schemes to achieve rates close to $C(\kappa)$.

We consider the following class of coding schemes. To write on cell $i$, the write controller  applies stimuli $X_{i}^{(1)}, X_{i}^{(2)}, \ldots$ until the output falls within a target region $T_i$, where $T_i$ is a subset of the output space. The $k$th write stimulus $X_i^{(k)}$ can depend on the outputs of the previous stimuli, denoted $Y_i^{(1)}, \ldots, Y_i^{(k-1)}$. Formally, a rewrite code of rate $R$ over an array of $n$ cells is defined by:
\begin{itemize}
\item An encoder mapping which maps a message in $\{1,\ldots,2^{nR}\}$ to a sequence
$((X_1,T_1), \ldots, (X_n,T_n))$, where $T_i$ is the target region for cell $i$, and $X_i =(X_{i}^{(1)}, X_{i}^{(2)},\ldots)$ is the input strategy for cell $i$.
\item A decoder which maps the output sequence $(Y_1, \ldots Y_n)$ to $\{1,\ldots,2^{nR}\}$.
\end{itemize}

 For cell $i$, the number of writes needed for the output to fall within  region $T_i$ is a random variable, denoted $\tau_i(X_i, T_i)$, where
 $X_i = (X_{i}^{(1)},X_{i}^{(2)}, \ldots)$ is the input strategy.
The average write-cost of the code is $\frac{1}{n}\sum_{i=1}^n \expec\tau_i(X_i,T_i)$. Due to the statistical independence of the cells, the capacity for an average cost constraint $\kappa$ is \cite{Rew08,UnifCap10}
\be C(\kappa)= \sup_{X,T: \expec\tau(X,T) \leq \kappa} I(XT; Y). \label{eq:gen_cap_formula} \ee

The capacity formula in \eqref{eq:gen_cap_formula} is not easy to compute in general. This is because the optimization is over adaptive strategies where each input stimulus can depend on the outcomes of the previous stimuli. Adaptive strategies are particularly useful in channels with hidden state because we get a better estimate of the state with each write, which can be used to generate the next stimulus.\footnote{For memoryless rewritable channels without state, we can restrict the input strategies to be non-adaptive, i.e. repeatedly apply the same stimulus to a cell until the target region is hit. See \cite{UnifCap10,RewJour12}.} For intuition, consider two extreme cases:
\begin{itemize}
\item When $\kappa=1$, we are allowed only one write attempt and the hidden state is treated as an additional noise variable.

\item When the average cost constraint $\kappa \to \infty$, we can spend a number of write attempts to get a very good estimate of the state, and use the remaining writes to store information by designing the input stimulus to nullify the effect of the state. Thus we expect the storage rate to approach the no-state capacity when $\kappa$ is very large.
\end{itemize}
For $1<\kappa<\infty$, the challenge is to simultaneously learn the state while attempting to store information at a high rate.

The main contributions of this paper are as follows.
\begin{enumerate}
\item In Section \ref{sec:gen_lb}, we derive a capacity lower bound for continuous-output rewritable channels with state.  The scheme used to obtain this bound involves state estimation phase followed by a coding phase. The writing strategy in the coding phase combines two techniques from multi-user information theory: Gelfand-Pinsker coding  \cite{GelfPin80, costa83}  and superposition coding \cite{CoverBC98}. The AWGN rewritable channel is discussed as an example.

    \item In Section \ref{sec:unif_channel}, we focus on the uniform noise channel and present a coding scheme that is computationally simple and amenable to practical implementation. The scheme implicitly combines state estimation and coding, and is shown to be asymptotically optimal as the number of rewrites gets large.
\end{enumerate}

The rewritable channel considered in this paper is a stylized model relevant to technologies like Phase Change Memory and Resistive RAM which have analog outputs. Both these memory technologies are known to be affected by variability across devices \cite{PCM10,lee2010evidence}, which to the first order can be modeled as a hidden state. Though relaxing assumptions such as noiseless reads would make the model more realistic,  we believe that the current model gives useful insights regarding how rewrites can be harnessed to improve the storage density of these memories.

\emph{Notation}: We use upper-case letters to denote random variables and bold-face notation for random vectors.  Entropy and mutual information are measured in bits, and logarithms are with base $2$ unless otherwise mentioned.

\section{Lower Bound  on the Rewrite Capacity} \label{sec:gen_lb}

 For an average write cost  $\kap$, we design a scheme consisting of two phases: an estimation phase of $l$ ($< \kap$) writes  to learn the state $S$, and a coding phase requiring an average of $\kap - l$ writes. For the coding phase, we combine two techniques: 1) Gelfand-Pinsker coding \cite{GelfPin80} which achieves the optimal rate given the state estimate if we are allowed only a single write for the coding phase (i.e., $\kappa -l = 1$), and 2) Superposition coding \cite{CoverBC98} to store an additional $\log(\kap-l)$ bits/cell when $\kap > l+ 1$.

Before presenting the general result, we describe the coding scheme for the AWGN rewritable channel to highlight the main ideas.

\subsection{The AWGN channel with hidden state} \label{subset:awgn}

The channel is defined by \eqref{eq:unif_state_model} with the write noise $W \sim \mc{N}(0,N)$ and the state $S \sim \mc{N}(0, \sigma^2_s)$. We assume that there is an average power-constraint $P$, i.e., in each write attempt the average power of the input stimulus across the $n$ cells is at most $P$.

\emph{State Estimation}: The first step is to construct an estimate of the state of each cell using  $l$ writes. Due to the symmetry of the channel model, this can be done by applying any input, say $c$,  for $l$ writes and recording the outputs $Y^{(1)}, \ldots, Y^{(l)}$ which are generated according to
\be
\begin{split}
Y^{(1)} & = c + W^{(1)} + S \\
&\vdots\\
Y^{(l)} & = c + W^{(l)} + S
\end{split}
\label{eq:state_est}
\ee
where $W^{(1)}, \ldots, W^{(l)}$ are independent $\mc{N}(0,N)$ random variables. The minimum-mean squared error (MMSE) estimate of $S$ given the observations $Y^{(1)}, \ldots, Y^{(l)}$  is
\be
\hat{S}(l) = \expec[ S \mid Y^{(1)}, \ldots, Y^{(l)}]  = \frac{\sigma^2_s}{l \sigma^2_s + N} \sum_{j=1}^l  (Y^{(i)} - c).
\ee

\emph{Encoding}:
The write channel for the $(l+1)$th write can be expressed as
\be
Y^{(l+1)} = X^{(l+1)} + \hat{S}(l) + (S - \hat{S}(l)) + W^{(l+1)}.
\label{eq:est_chmodel}
\ee
The estimate  $\hat{S}(l)$  is known to the encoder prior to the $(l+1)$th write. Further, $(S - \hat{S}(l)) $ is independent of $\hat{S}(l)$ due to the orthogonality principle \cite{Kayv1} and the joint Gaussianity of $(S, \hat{S}(l))$.

 Let us first consider the case where  we use only a single write after the estimation period. For write $l+1$, \eqref{eq:est_chmodel}  describes a channel with state
 $\hat{S}(l)$ known to the encoder and  effective channel noise $S - \hat{S}(l)+ W^{(l+1)}$ which is independent of   $\hat{S}(l)$. The effective channel noise is a Gaussian random variable distributed as $\mc{N}(0, N_{\text{eff}, l})$ where
\be
\begin{split}
N_{\text{eff}, l} & = \expec[(S - \hat{S}(l) + W^{(l+1)})^2]  \\
& = \expec[(S - \hat{S}(l))^2] + \expec[(W^{(l+1)})^2]  = \frac{\sigma_s^2 N}{l \sigma_s^2 +N} + N.
\end{split}
\label{eq:eff_noise_var}
\ee
 The optimal coding scheme for this channel is the `writing on dirty-paper' scheme of Costa \cite{costa83}. The key idea is to incorporate part of the known state $\hat{S}(l)$ into the codeword. This is done by building a codebook over an auxiliary random variable $U \sim \mc{N}(0,  P+ \alpha^2 \sigma^2_{s,l} )$ where
 \begin{align}
 \alpha & = \frac{P}{P + N_{\text{eff}, l}},  \label{eq:alph_def} \\
\sigma^2_{s,l}  & =  \expec[\hat{S}(l)^2] = \frac{l \sigma_s^4}{l \sigma_s^2 + N} . \label{eq:sig_l}
 \end{align}
 
 Let the storage rate be $R$ bits/cell. We build a $U$-codebook with $2^{nR_1}$ codewords whose elements  are generated i.i.d 
 according to $\mc{N}(0,  P+ \alpha^2 \sigma^2_{s,l} )$. The value of $R_1 > R$ will be specified below.  The codebook is divided into $2^{nR}$ bins, with  each bin containing $2^{n(R_1-R)}$ codewords. Each bin represents a message in the set  $\{1, \ldots, 2^{nR}\}$. To transmit message $m$, the encoder attempts to find a codeword $\mathbf{U}$ {within} bin $m$ such that $(\mathbf{U} - \alpha \mbf{\hat{S}}(l) )$ is nearly orthogonal to
$\mbf{\hat{S}}(l)$. Formally, the encoder finds a codeword $\mathbf{U}$ that is jointly typical\footnote{Roughly speaking, sequences $(\mathbf{U}, \mbf{S})$ are jointly typical according to distribution $P$ if their empirical joint distribution is close to  i.i.d.  $P$.} \cite{CoverT} with  $\mbf{\hat{S}}(l)$  according to the distribution described by
\be U = X  + \alpha \hat{S}(l) \label{eq:u_gp} \ee
 where $X \sim \mc{N}(0, P)$ and $\hat{S}(l) \sim \mc{N}(0, \sigma^2_{s,l} )$ are independent Gaussians. From rate-distortion theory, this step will be successful if the number of sequences in each bin  $2^{n(R_1-R)}$ is  larger than $2^{n I(U; \hat{S}(l))}$, where the mutual information computed using the joint distribution described by \eqref{eq:u_gp}. The codeword written on the $n$ cells is
 \be \mbf{X}^{(l+1)} = \mathbf{U}   - \alpha \mbf{\hat{S}}(l).  \ee
Note that $\mbf{X}^{(l+1)}$ has average power nearly $P$, due to \eqref{eq:u_gp}. The sequence stored on the cells is
\be
\mbf{Y}^{(l+1)} = \mbf{X}^{(l+1)}  + \mbf{S} + \mbf{W}^{(l+1)}.
\label{eq:ch_seqs}
\ee

The decoder's task is to decode the codeword $\mbf{U}$ from $\mbf{Y}^{(l+1)}$. The corresponding bin index then gives the message.
If the encoding operation is successful,  $(\mbf{U},   \mbf{Y}^{(l+1)})$ are jointly typical according to
\be
\begin{split}
U & = X  + \alpha \hat{S}(l) ,  \\ Y & = X + \hat{S}(l) + (S-\hat{S}(l) +W),
\end{split}
\label{eq:ch_out}
\ee
where $X \sim \mc{N}(0,P)$, $\hat{S}(l) \sim \mc{N}(0, \sigma^2_{s,l})$, and $(S-\hat{S}(l) +W)  \sim \mc{N}(0, N_{\text{eff}, l})$ are mutually independent random variables.
If we use a maximum-likelihood or joint typicality decoder \cite{CoverT}, the codeword $\mathbf{U}$  can be successfully decoded if $R_1 < I(U;Y)$.

Combining this with the earlier bound $R_1-R > I(U;\hat{S}(l))$, we conclude that rates
\be R < I(U;Y) - I(U;\hat{S}(l)) \ee
are achievable. Computing this with the joint distribution specified in  \eqref{eq:ch_out}, we obtain  that any rate
 \be
 R < \frac{1}{2} \log\left(  1 + \frac{P}{N_{\text{eff}, l}}  \right)
 \label{eq:max_1write}
 \ee
is achievable.

 If we restrict ourselves to a single write after the estimation period, \eqref{eq:max_1write} gives the optimal rate. This is because even when the encoder and decoder  \emph{both} know $\mbf{\hat{S}}(l)$ a priori, the maximum rate is given by \eqref{eq:max_1write}.  (The decoder can simply cancel off any effect of  $\mbf{\hat{S}}(l)$ from the stored value.)  When $\mbf{\hat{S}}(l)$ is not available at the decoder, the Costa coding scheme nullifies its effect by incorporating part of it into the codeword $\mathbf{U}$.

\begin{figure}[t]
\flushleft
\includegraphics[width=3.5in]{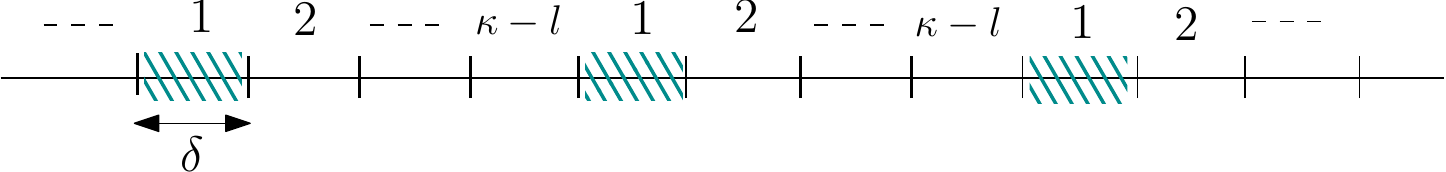}
\caption{\small{Target regions for superposition coding. The shaded regions together form the target region for one message.}}
\vspace{-5pt}
\label{fig:superpos}
\end{figure}

\emph{Superposition}: When $\kappa > l +1$, we have available more than one write after the estimation period. We use  superposition coding  to store additional bits using the remaining  writes. For the sake of intuition, temporarily assume that $\kappa$ is an integer.

The idea  is to partition the output space  ($\mathbb{R}$ for the AWGN channel)  into $(\kappa - l)$ different regions such that the output  is equally likely to fall in each of these regions \emph{regardless} of the input. This is done in the following way. We divide the real line into intervals of length $\delta$, and assign to the intervals labels $1, \ldots, (\kappa-l)$ in succession, as shown in Figure \ref{fig:superpos}.  Target region $1$ is the union of the intervals marked $1$, target region $2$ is the union of the intervals marked $2$, and so on. Formally, we define the target regions $T_j$ for $j=1, \ldots, \kappa-l$ as
\be
T_i= \cup_{i \in \mathbb{Z}} \Big[  ((\kappa -l)i+j)\delta - \delta, \quad ((\kappa -l)i+j)\delta  \Big).
\ee
We let $\delta \to 0$ for reasons explained below.

For each cell $i \in \{ 1, \ldots,n \}$,  the additional information stored is represented by a message $m_i$ drawn uniformly from the set $\in \{1, \ldots,  \kappa -l \}$. At the end of estimation period, the controller uses the state sequence estimate $\hat{\mbf{S}}(l)$ to determine the codeword $\mbf{U} = (U_1, \ldots, U_n)$ of the Costa scheme. It then repeatedly applies stimulus  $X_i$ to cell $i$ until the output falls in the target region $m_i$. Recall that
\[ X_i = U_i - \alpha \hat{S}_i(l).  \]

 In each write attempt, the noise realization is an independent realization of a $\mc{N}(0, N)$ random variable. However, the state estimation error $S-\hat{S}(l)$  remains \emph{constant} across attempts and is unknown to the decoder. Defining each target region as a collection of disjoint infinitesimal intervals ensures that  the output in each write attempt is equally likely to lie in each of the $(\kappa -l)$ target regions, regardless of the value of $S-\hat{S}(l)$. Thus the number of writes required to obtain an output in the desired region $m_i$ is a geometric random variable with mean $(\kappa -l)$.  
 
 The total number of writes (including the estimation period) for cell $i$ is denoted $\tau_i$ and the final value stored in cell $i$ is $Y_i^{(\tau_i)}$. The discussion above shows that 
 \be  \expec[\tau_i] = l + \kappa - l = \kappa. \ee 
 
\emph{Decoding}: The decoder observes the stored sequence
\[ \mbf{Y}^{(\tau)} = ( Y_1^{(\tau_1)}, \ldots, Y_n^{(\tau_n)})\]
and attempts to decode the codeword $\mbf{U}$. The key observation is that $(\mbf{U},  \mbf{Y}^{(l+1)})$  and $(\mbf{U},  \mbf{Y}^{(\tau)})$ have the same joint distribution. This is because  for each cell $i$, the write stimulus and the channel state remain the same for writes $(l+1)$ through $\tau_i$, and the noise realizations  $W^{(l+1)}, \ldots, W^{(\tau_i)}$ are  i.i.d. $\mc{N}(0,N)$. Thus $(\mbf{U},  \mbf{Y}^{(\tau)})$ is jointly typical according to \eqref{eq:ch_out}, and the codeword $\mbf{U}$ can be reliably decoded if $R$ satisfies
\eqref{eq:max_1write}. The target region containing the output of  cell $i$ directly gives the message  $m_i \in \{1, \ldots, \kappa-l \}$ stored in the superposition phase.

When $\kappa$ is not an integer, we can vary the number of target regions across the cells in order to achieve a write-cost of $\kappa$. For example, if
$\kappa =1 + \lambda$, for $\lambda \in (1,2)$, we can code a fraction $\lambda$ of the $n$ cells at  average cost  $2$ and the remaining $n(1-\lambda)$ cells at average cost $1$ to obtain an overall cost of $\kappa= 1+ \lambda$.  Thus the straight line joining the value of  the lower bound at $\kappa=1$ and $\kappa=2$  is a lower bound for $\kappa \in (1,2)$. In general, the convex hull of the rates achieved  at the integer points can be achieved through `cost-sharing'  between cells.

The performance achieved of two-step coding strategy described above is summarized in the following proposition.
\begin{prop}
Consider the channel described by \eqref{eq:unif_state_model}  with state $S \sim \mc{N}(0,\sigma_s^2)$, noise $W \sim \mc{N}(0,N)$ and an average power constraint $P$ on the input. With average cost $\kappa \geq 1$, the rewrite capacity satisfies
\be
\begin{split}
&C(\kappa)  \geq \\
&\text{conv} \left( \max_{l \in \{0,1, \ldots, \lfloor \kappa \rfloor -1 \} }
\left\{  \frac{1}{2} \log  \left( 1 + \frac{P}{N_{\text{eff}, l}}  \right)   + \log \lfloor \kappa - l \rfloor \right\} \right)
\end{split}
\label{eq:awgn_prop}
\ee
where conv denotes the convex hull and
\[  N_{\text{eff}, l} = N \left( 1 + \frac{\sigma_s^2 }{l \sigma_s^2 +N}\right).\]
\label{prop:awgn}
\end{prop}

\begin{figure}[t]
\flushleft
\includegraphics[height=2.3in]{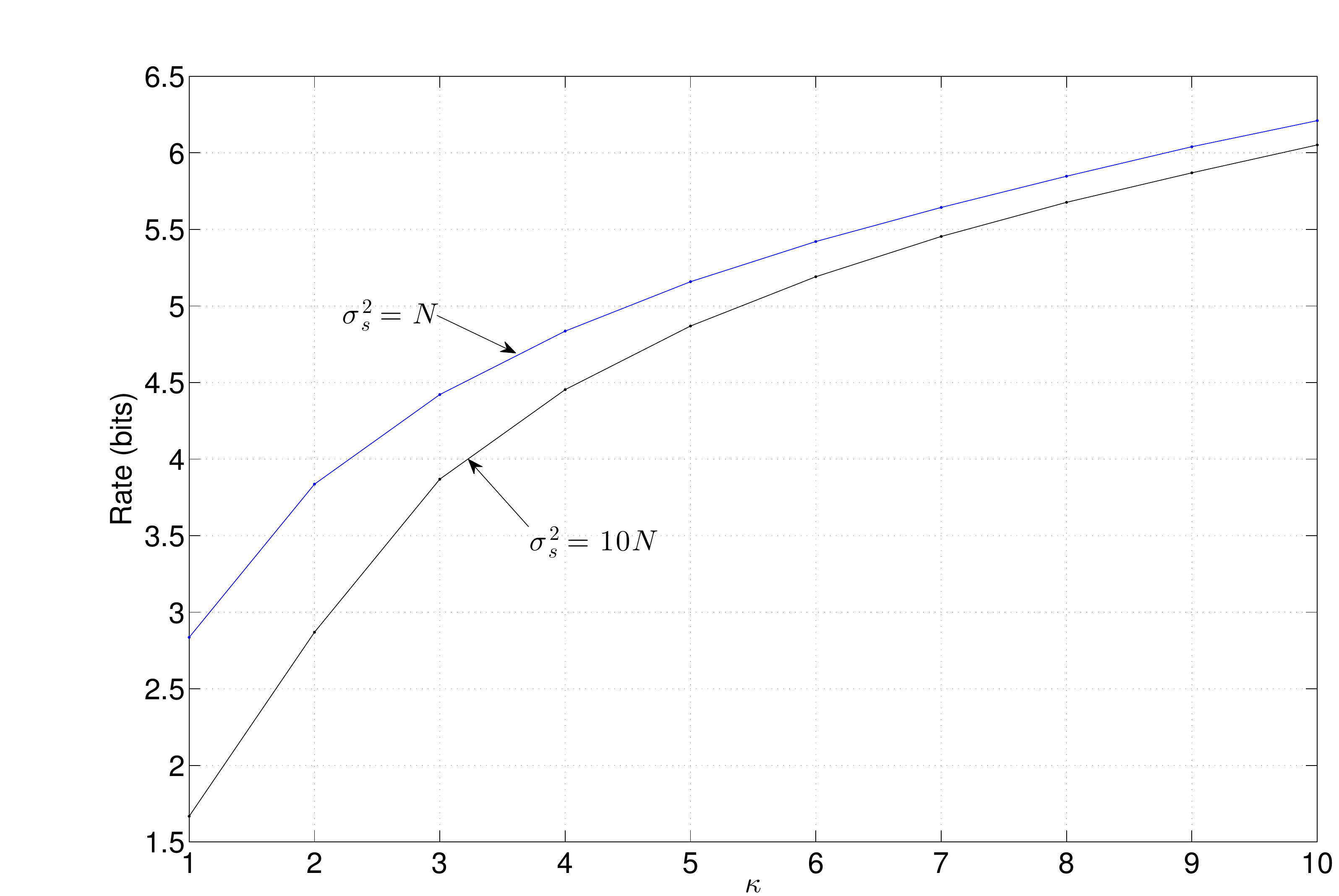}
\caption{\small{ Lower bound of Proposition \ref{prop:awgn} for $\tfrac{P}{N}=100$. The top curve is for $\sigma_s^2 = N$ and the bottom one for $\sigma_s^2 = 10N$.}}
\vspace{-5pt}
\label{fig:awgn}
\end{figure}

Figure \ref{fig:awgn} shows the capacity lower bound for $\tfrac{P}{N}=100$ and $1\leq \kappa \leq 10$. Curves are plotted for $\sigma_s^2=N$ and $\sigma_s^2=10N$.
For the second case,  the maximum in \eqref{eq:awgn_prop} is attained with an estimation period of $l=0$ for $\kappa=1$,  $l=1$ for $2 \leq  \kappa \leq 9 $, and  $l=2$ for $\kappa = 10$.

\subsection{Lower bound for General Channels}

We consider channels whose  output support is continuous valued, i.e., for $\forall (x, s) \in \mc{X} \times \mc{S}$, $P_{Y|XS}(. | x, s)$ is absolutely continuous with respect to the Lebesgue measure. This assumption is necessary because the definition of target regions for superposition coding in Section \ref{subset:awgn} implicitly assumes continuous valued outputs. The superposition idea can be extended to many discrete channels as well, but we do not pursue this here in order to keep the exposition simple.

 For a channel with average write cost $\kappa$, the two-step strategy involves: a) Designing a suitable estimator to estimate the state sequence using
$l$ writes, and b)  storing information in the remaining $(\kappa - l)$ writes using Gelfand-Pinsker coding and superposition. The Costa coding scheme for the AWGN channel is a special case of the Gelfand-Pinsker scheme for general memoryless channels with state, with the state known a priori  at the encoder.

\begin{thm}
Consider a channel  $P_{Y|XS}(. | x, s)$ that is absolutely continuous with respect to the Lebesgue measure for all $(x, s) \in \mc{X} \times \mc{S}$.  With average cost
$\kappa \geq 1$, the capacity satisfies
\be
\begin{split}
&  C(\kappa) \geq \\
& \text{conv} \left( \max_{l \in \{0, \ldots, \lfloor \kappa \rfloor -1\}} \hspace{-4pt} \max_{\mc{P}}    \left\{ I(U;Y) - I(U; \hat{S}(l))  +    \log \lfloor \kappa - l \rfloor \right\} \right)
\end{split}
\ee
 where  $\mc{P}$ is the set of joint distributions of $(S, \hat{S}(l), U, X, Y)$ of the form
\[
P_{S} \cdot P_{\hat{S}(l)|S} \cdot P_{U|\hat{S}(l)} \cdot \mathbf{1}_{ X=f(U, \hat{S}(l)) } \cdot P_{Y|X S }.
\]
\label{thm:cont_ch_lb}
\end{thm}
\proof See Appendix.

\emph{Remarks}:
\begin{enumerate}
\item In the set of joint distributions $\mc{P}$ the state distribution $P_S$ and the channel law $P_{Y|XS}$ are fixed. The maximization over $\mc{P}$  is  therefore over the choice of estimator $\hat{S}(l)$,  auxiliary distribution $P_{U|\hat{S}(l)}$, and function $f$ to generate the channel input $X$ from $(U, \hat{S}(l))$.

\item The MMSE estimator is optimal for the AWGN average-power constrained channel, but in general the optimal estimator depends on the channel law and the input constraints.
\end{enumerate}

We conclude this section with a brief discussion of the shortcomings of the two-step coding scheme discussed in this section.
First, dedicating $l$ writes to estimating the state and then coding is not optimal in general. A scheme that simultaneously performs estimation and coding in each write attempt is likely to yield higher rates, but such a scheme may also be harder to analyze.

The information is stored in the cell array in coded in two ways: through Gelfand-Pinsker coding and  superposition coding. Each of these poses a different challenge for practical implementation. In the Costa/Gelfand-Pinsker scheme we used joint typicality or maximum-likelihood decoding,  both of which are computationally infeasible for  a large array of $n$ cells. For the AWGN case, there has been has been progress towards feasible decoders using structured codebooks such as those based on lattices \cite{erezten05,zamir02nested}.

For superposition coding, we need the reads to be very accurate as the width $\delta$ of the intervals is made small (cf.  Figure \ref{fig:superpos}). This is important during encoding (so that the controller knows when to stop writing) as well as decoding (for the decoder to know which target region the output lies in).  This problem can be handled by using an outer error-correcting code to correct errors that arise due to imperfect reads.

In the following section, we design a coding scheme that addresses all the above issues for the uniform-noise channel.

\section{Uniform Noise Channel with Hidden State} \label{sec:unif_channel}

The channel is described by \eqref{eq:unif_state_model}  with the write noise $W$ is uniformly distributed in $[-a/2,a/2]$,  and the state $S$ is uniformly distributed in $[0,B]$.
$a,B$ are assumed to be known positive constants. For ease of analysis, we will assume that $B < a$.

We present two code constructions, each of which gives a lower bound on the rewrite capacity. The first is sub-optimal but gives insight into features of good coding strategies. The second construction yields a better lower bound which is asymptotically optimal, i.e., it is arbitrarily close to the no-state capacity for sufficiently large cost constraint. The second scheme implicitly performs simultaneous state estimation and coding; further, it is computationally simple and robust to small inaccuracies in the reading process.

\subsection{Code Construction $1$} \label{sec:nostate}
 For  the uniform noise rewrite channel without state  given by $Y=X+W$, the basic coding idea is that with an average of $\kappa$ rewrites, we can shrink the effective width of the noise interval to $a/\kappa$. The average-cost  capacity was obtained in \cite{UnifCap10}.
\begin{fact}  \cite{UnifCap10}
For $ \kappa \geq \kappa_0 \triangleq {\lceil \frac{1+a}{a} \rceil}/ \left(\frac{1+a}{a}\right)$, the rewrite capacity with  average cost constraint $\kappa$ is
\[ C(\kappa) = \log \left(\frac{1+a}{a} \kappa \right).\]
\label{fact:nostate_cap}
\end{fact}
When $\frac{1+a}{a} \kappa$ is an integer, the capacity is achieved by simply dividing the space $[-a/2, 1+a/2]$ into equal-sized intervals of length $a/\kappa$ and choosing the target region $T$ to be one of these intervals with equal probability. The input $X$ is any point which maximizes the probability of the output falling in the region $T$. When $\frac{1+a}{a} \kappa$ is not an integer, the capacity is achieved by  a careful generalization of the above idea, described in \cite{UnifCap10}.

When there is an unknown state offset $S \in [0,B]$, the idea is to define each target region such that there is exactly one subset of width $a/\kappa$ that can be accessed with a \emph{fixed} input and an average of $\kappa$ writes, irrespective of the offset.

\begin{figure}[t]
\centering
\includegraphics[height=1.6in]{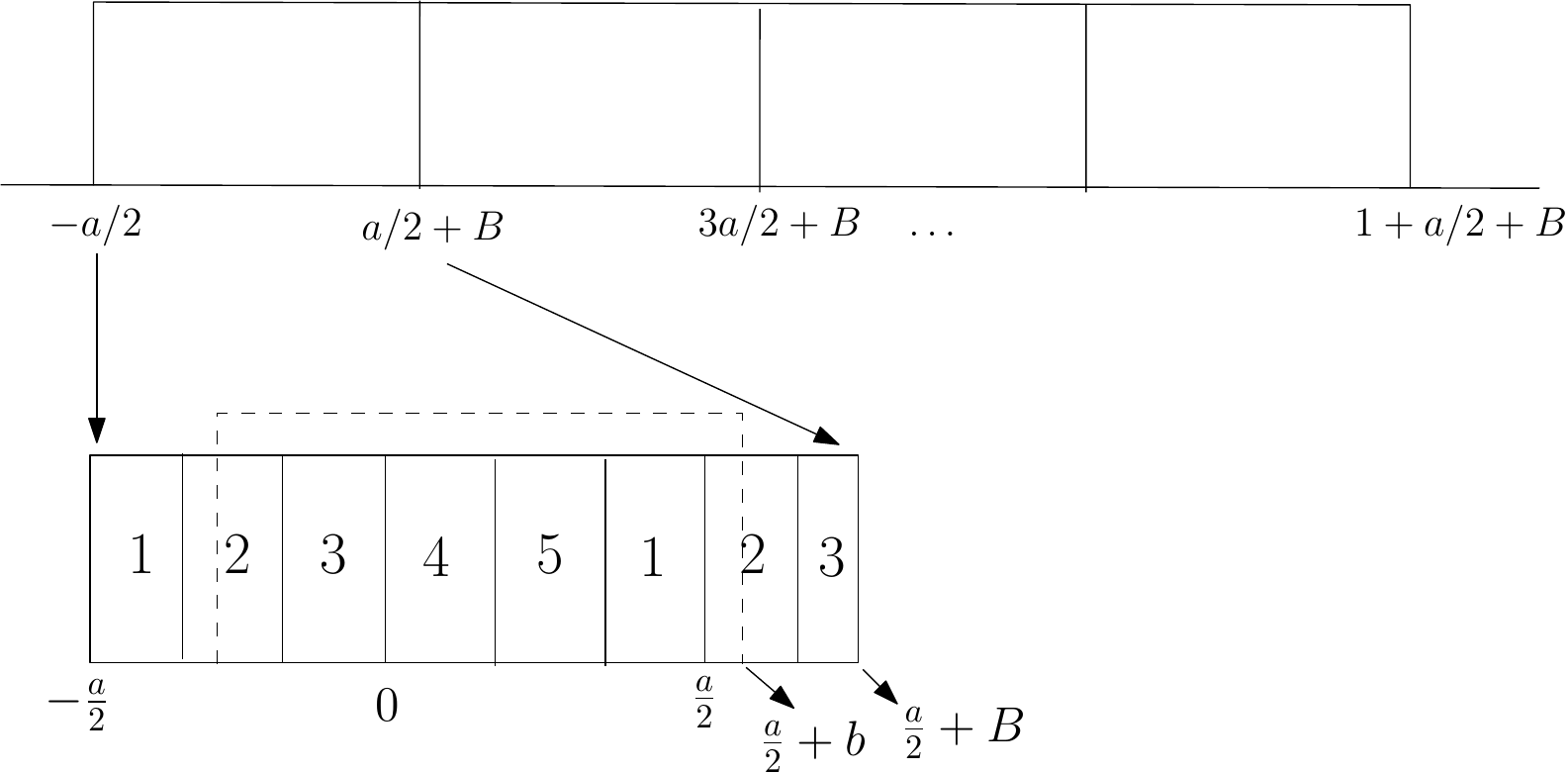}
\caption{\small{Each interval of width $a+B$ is divided into $\kappa=5$ target regions. The target regions in the interval $(-a/2,a/2+B)$ are shown.
To write on a  region in this interval, the stimulus $0$ is applied. The dashed lines indicate the part of the interval accessible with $S=b$.}}
\vspace{-10pt}
\label{fig:scheme1}
\end{figure}

\begin{figure*}[t]
\centering
\includegraphics[width=5in]{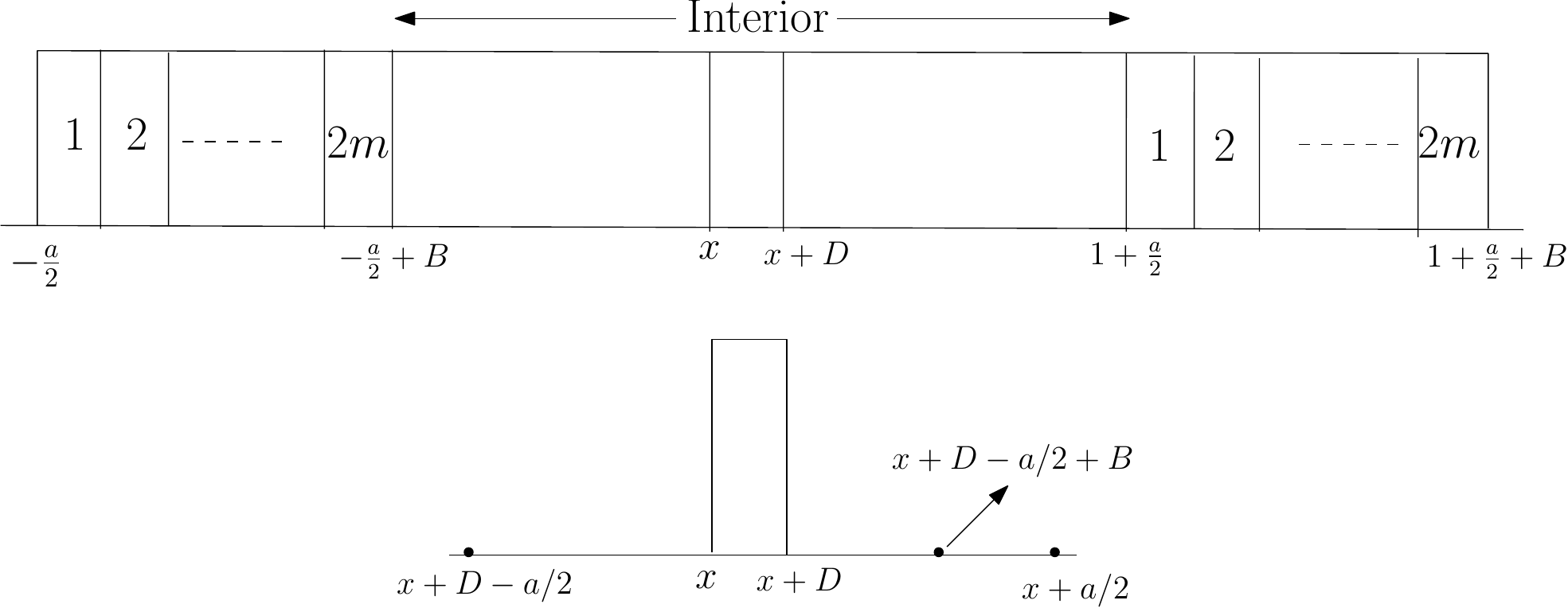}
\caption{\small{Construction $2$: The output space is divided into interior and exterior target regions. The bottom figure shows the an interior target region $[x,x+D]$. The  state shifts the input  $x+D-a/2$  to the right by an amount at most $B$. For all $b\in [0, B]$, the probability of hitting the target region in each write attempt is $D/a$ as long as $D<a-B$.}}
\vspace{-10pt}
\label{fig:scheme2}
\end{figure*}

\begin{prop}
For the uniform noise rewrite channel with hidden state and average cost $\kappa \geq 2$,
\be
C(\kappa) \geq \log \left(\kappa \left\lfloor \frac{1+a+B}{a+B} \right\rfloor \right).
\label{eq:scheme1}
\ee
\label{prop:scheme1}
\end{prop}

\begin{IEEEproof}
\emph{The target regions}:
Refer Figure \ref{fig:scheme1}. The output space $[-a/2, 1+a/2+B]$ is first divided into intervals of length $(a+B)$ each. There are $N =\lfloor \tfrac{1+a+B}{a+B} \rfloor$ such intervals, denoted $\mc{Z}_i, 0\leq i \leq N-1$.  If $\tfrac{1+a+B}{a+B}$ is not an integer, the remaining output space $(N(a+B), 1+\tfrac{a}{2}+B]$  is discarded.

For clarity, consider the case where $\kappa$ is an integer. Divide each interval $\mc{Z}_i$ into $\kappa$ target regions.
The target regions for the first interval are defined as shown in Figure \ref{fig:scheme1}.  Target region $1$ is the interval
$[-\tfrac{a}{2}, -\tfrac{a}{2} + \tfrac{a}{\kappa}] \cup [\tfrac{a}{2}, \tfrac{a}{2} + \tfrac{a}{\kappa}]$;  region $2$ is the interval $[-\tfrac{a}{2}+\tfrac{a}{\kappa}, -\tfrac{a}{2} + \tfrac{2a}{\kappa}] \cup [\tfrac{a}{2}+\tfrac{a}{\kappa}, \tfrac{a}{2} + \tfrac{2a}{\kappa}]$, and so on.  Similarly, $\kappa$ target regions are defined for each of the $N$ intervals.

\emph{Encoding}: To reach target region $t$ in interval $\mc{Z}_i$ for $t \in \{1, \ldots , \kappa \}, i \in \{0, \ldots, N -1\}$, apply input
$X=(a+B)i$ until the output falls within region $t$ in interval $i$.  With this input, the accessible part of the target region has width exactly $a/ \kappa$ for any value of $S \in [0,B]$. This is illustrated in the bottom part of Figure \ref{fig:scheme1}. Regardless of the offset, the probability of the output falling within the target region on any write attempt is $\frac{a/\kappa}{a}$. The average number of rewrites is therefore $\kappa$.

The total number of target regions is $N\kappa$ and by assigning them equal probability, the rate is
\be
\begin{split}
I(XT;Y) = I(T;Y) =H(T) = \log (N\kappa)
\end{split}
\ee
where the first equality holds because the input $X$ is a function of the target region $T$, and the second equality is due to $T$ being uniquely determined by $Y$.

The general case where $\kappa$ is not an integer can be handled by an extension of the above scheme using the techniques in \cite{UnifCap10}.
\end{IEEEproof}

\emph{Remark}:  When $\tfrac{1+a+B}{a+B}$ is an integer, \eqref{eq:scheme1} can be written as
\be
C(\kappa) \geq \log \left(\kappa \frac{1+a}{a}\right) - \log \left( \frac{1+  B /a}{1+ B/(1+a)}  \right)
\label{eq:two_term_scheme1}
 \ee
The first term  above is the capacity when $S=0$, or when $S$ is precisely known at the encoder. The second term is the loss incurred by the coding scheme due the state being unknown.

In the above construction, we designed the target regions so that each one can be accessed with equal probability regardless of the value of $S$.  We did not use the rewrites to do any state estimation. The sub-optimality of this strategy is seen by observing that even when the number of rewrites $\kappa$ is very large, the lower bound of Proposition \ref{prop:scheme1} is strictly less than the  capacity with $S=0$,  given by the first term in \eqref{eq:two_term_scheme1}.
 The next construction  remedies this deficiency.
\subsection{Code Construction $2$}
 As shown in Figure \ref{fig:scheme2}, divide the output space $[-a/2, 1+a/2+B]$ into two regions: the interval $[-a/2+B, 1+a/2]$ called the `interior',
and the remaining space $[-a/2, -a/2+B] \cup [1+a/2, 1+a/2+B]$ called the `exterior'.

\emph{Interior target regions}: Divide the interior into intervals (target regions) of width $D$. The key observation is that  if $D < a-B$, regardless of the value of $S$,  each interior target region is fully accessible  with an average of $a/D$ write attempts with a {fixed} input.
As illustrated in the bottom part of Figure \ref{fig:scheme2}, to access the interior target region $[x, x+D]$, apply the stimulus $(x+D-a/2)^+$. To fully access the region with offset $b$, we need
\[ (x+D-a/2)^+ + b - a/2 < x  \]
which holds for all $b \in [0,B]$ as long as $D<a-B$.

\emph{Exterior target regions}: As shown in Figure \ref{fig:scheme2},  define $2m$ exterior target regions for an integer $m\geq1$. For $i=1,\ldots,2m$,
the $ith$ exterior target region, labeled $E_i$, is
\ben
\begin{split}
&\left[-\frac{a}{2}+(i-1)\frac{B}{2m}, \ -\frac{a}{2}+ i\frac{B}{2m}\right]\\
 & \quad \bigcup \left[1+ \frac{a}{2} +(i-1)\frac{B}{2m}, \ 1+\frac{a}{2}+ i\frac{B}{2m}\right]. \end{split}
\een

With this construction, we present a coding scheme that achieves the following lower bound on the rewrite capacity.
\begin{thm}
For $\kappa \geq 2$
\[ C(\kappa) \geq \max_{p,D,m} \ h(p) + p \log  \left\lfloor \frac{1+a-B}{D} \right\rfloor  + (1-p)\log 2m \]
where the maximum is over $p \in [0,1]$, $D \in (0,a-B)$ and integers $m \geq 1$ that satisfy
\be
\begin{split}
& (1-p) \frac{a}{B}\left[2m+1 + \frac{1}{m}\sum_{i=1}^{m} \left( \ln\frac{i-\delta_i}{(1-\delta_i)^2}  + \frac{\delta_i}{1-\delta_i} \ln \delta_i \right) \right] \\
& \quad + p \frac{a}{D} \leq \kappa
\end{split}  \label{eq:pd_ext_cond_imp}
\ee
where  the optimal $\delta_i \in [0,1)$ for $i=1,2,\ldots,m$  is determined by the following equation:
\be
2(1-\delta_i)^2 + 3(i-1)(1-\delta_i) + (i-\delta_i) \ln \delta_i = 0.
\label{eq:opt_deli_thm}
\ee
The optimal $\delta_i$ for a few values of $i$ are listed in Table \ref{tab:deli}.
\label{thm:ext_scheme}
\end{thm}

\emph{Remark}:
The $\delta_i$ in the above theorem can chosen to be arbitrary values in [0,1). Picking $\delta_i$ that satisfy \eqref{eq:opt_deli_thm} minimizes the number of rewrites, given by the left side of \eqref{eq:pd_ext_cond_imp}. For example, \eqref{eq:pd_ext_cond_imp} can be replaced by a  simpler condition obtained by setting $\delta_i=0$ for all $i$:
\be p \frac{a}{D} + (1-p) \frac{a}{B}\left(2m +1+ \frac{\ln m!}{m}\right) \leq \kappa. \label{eq:pd_ext_cond} \ee

The proof of the theorem is given in the next section.
Figure \ref{fig:plot} shows the lower bound of Theorem \ref{thm:ext_scheme} with $a=1/3$ and $B=a/2$ for various values of $\kappa$.

We now show that the above lower bound converges to the no-state capacity as the rewrite constraint $\kappa \to \infty$.
 \begin{corr}
 The rate $R(\kappa)$ achieved by Theorem \ref{thm:ext_scheme} satisfies
 \[ \left| \log \left(\frac{1+a}{a} \kappa \right) - R(\kappa) \right|  \to 0 \text{ as } \kappa \to \infty\]
\end{corr}
\begin{proof}
Choose $D=a/\kappa$ and $m=\frac{B}{2D}=\kappa \frac{B}{2a}$. Note that for all $B<a$, $D=a/\kappa < (a-B)$ for sufficiently large $\kappa$. With this choice and setting $\delta_i=0$ for all $i$,  the average number of rewrites given by the left-side of  \eqref{eq:pd_ext_cond} becomes
\[ p \kappa + (1-p) \frac{\kappa}{2} \left(2+ \frac{1}{m} + \frac{\ln m!}{m^2} \right) = \kappa(1 + \epsilon_{\kappa}), \]
where $\epsilon_{\kappa} =  O(\frac{\log \kappa}{\kappa})$ goes to $0$ as $\kappa \to \infty$.
Then with $p=\frac{1+a-B}{1+a}$, we see that
\[ R(\kappa(1+\epsilon_{\kappa})) = \log \left( \frac{1+a}{a} \kappa \right).\]
Therefore $R(\kappa) = \log \left( \frac{1+a}{a} \frac{\kappa}{1+\epsilon_{\kappa}} \right)$.
\end{proof}

\emph{Remarks} :
\begin{enumerate}
\item The coding scheme for the uniform noise channel (described in Section \ref{sec:proof}) stores information cell-by-cell, and therefore has low computational complexity. In contrast, each codeword in the Gelfand-Pinsker scheme of Section \ref{sec:gen_lb} is defined over a large array of $n$ cells, which makes the encoding and decoding computationally hard.

\item All results in this section generalize to the case where the hidden state uniformly distributed over a different support set of width $B$ that is different from $[0,B]$.

\item The coding scheme of construction 1 can directly be used when $B > a$. Construction $2$ needs modification -- the interior target region will not completely lie within the noise support when $B > a$. This can be addressed by shifting the input stimulus by an appropriate amount if it is determined that the offset $b>a$. This is  similar in principle to the switching strategy used for the exterior regions.
\end{enumerate}

\section{Proof of Theorem \ref{thm:ext_scheme}} \label{sec:proof}
To highlight the main ideas, we start with a simplified coding scheme for the case of $m=1$, i.e., two exterior target regions. 
\subsection{Two Exterior Target Regions}
\emph{Coding Scheme}: Fix  $p \in [0,1]$. For each cell, an interior target region is picked with probability $p$ and an exterior region is picked with probability $(1-p)$. All interior target regions are equally likely, as are the  exterior regions. Formally, each interior region has probability $p \tfrac{D}{1+a-B}$ and the two exterior regions each have probability $(1-p)/2$.
Refer Figure \ref{fig:two_regions}.
To write on  interior region $[x, x+D]$, repeatedly apply stimulus $(x+D-a/2)^+$ until the output falls within the region.

To write on exterior region $E_1$: Apply stimulus $1$ until either the output falls in $(1+a/2,1+a/2+B/2)$, or it is less than $1-a/2+B/2$. If the former occurs, stop. Otherwise, apply stimulus $0$ until the output falls in $(-a/2,-a/2+B/2)$. The intuition is that the right bin of $E_1$ is fully accessible with stimulus $1$  if the offset  lies in the interval $[B/2,1]$. We switch to the left bin of $E_1$ if we detect that the offset lies outside this interval.

To write on exterior region $E_2$: Apply stimulus $0$ until either the output falls in $(-a/2+B/2,-a/2+B)$ or it is greater than $a/2+B/2$. If the former occurs, stop. Otherwise, switch to applying stimulus $1$ until the output falls in  $(1+a/2+B/2,1+a/2+B)$. If the offset lies in the interval $[0,B/2]$, the left bin of $E_2$ is fully accessible with stimulus $0$. We switch to the right bin of $E_2$ if we detect that the offset lies outside this interval.

\emph{Analysis}: Since we have two exterior target regions with probability $(1-p)/2$, and $\lfloor \tfrac{1+a-B}{D} \rfloor$ interior regions each with probability
$p/ \lfloor \frac{1+a-B}{D} \rfloor$, the rate of information stored in each cell is calculated to be
\be H(T) = h(p) + p \log \lfloor \frac{1+a-B}{D} \rfloor + (1-p)\log 2. \ee
\begin{table}[t]
\caption{\small{Optimal value of $\delta_i$ for $1 \leq i \leq m$}}
\begin{center}
\vspace{-5pt}
\begin{tabular}{|c|c|c|c|c|c|c|}
\hline
$i$ & $1$ & $2$ & $3$ & $4$ & $5$ & $6$\\
\hline
$\delta_i$ & $0.2032$ & $0.1038$ & $0.0858$ & $0.0782$ & $0.0740$ & $0.0713$\\
\hline
\end{tabular}
\end{center}\label{tab:deli}
\vspace{-14pt}
\end{table}
Next we compute the average number of writes and set it equal to $\kappa$.
\be
\begin{split}
\kappa & = p \ \expec[\# \text{ writes} \mid \text{interior}] + (1-p) \expec[\# \text{ writes} \mid \text{exterior}] \\
& = p \frac{a}{D} + (1-p) \expec[\# \text{writes} \mid \text{exterior}].
\end{split}
\label{eq:rew_total}
\ee
By symmetry,
\be
\begin{split}
&\expec[\# \text{ writes} \mid \text{exterior}] = \expec[\# \text{ writes} \mid \text{ext. region } E_1]\\
&=   \int_{0}^{B} \expec[\# \text{ writes} \mid  E_1, S=b] \frac{1}{B} db\\
& = \int_{0}^{B/2} \expec[\# \text{ writes} \mid  E_1, S=b] \frac{1}{B} db  +  \int_{\frac{B}{2}}^{B} \frac{a}{B/2} \frac{1}{B} db
\end{split}
\label{eq:rew_ext_basic}
\ee
since the right bin of $E_1$ is fully accessible with stimulus $1$ when $S \in [B/2,B]$. We now show that for all $b \in [0,B/2)$,
\be
\expec[\# \text{ writes} \mid  E_1, S=b] = {4a}/{B}.
\label{eq:switch_avg}
\ee
\begin{figure}[t]
\includegraphics[width=3.5in, height=2.5in]{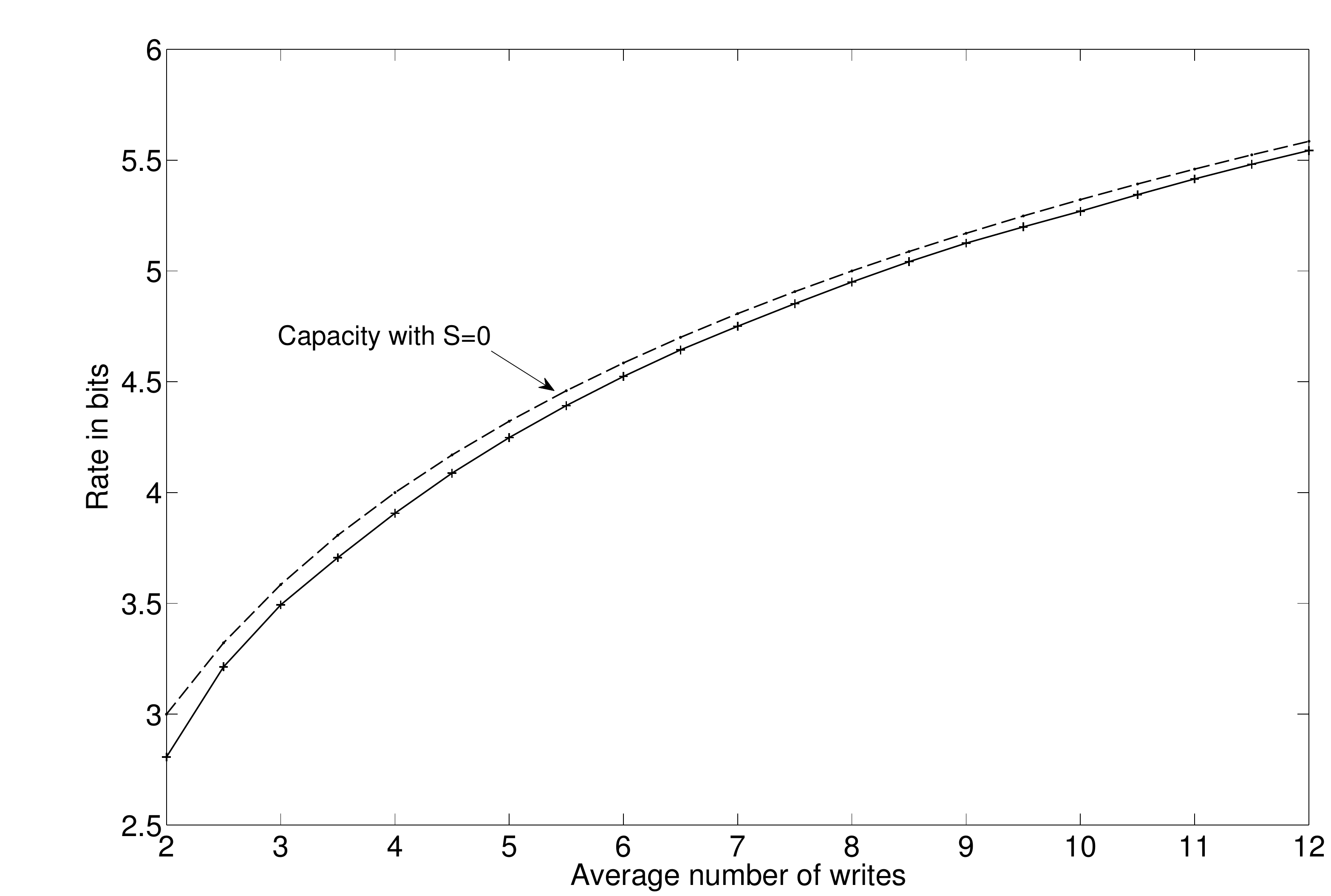}
\caption{\small{Achievable rate of Theorem \ref{thm:ext_scheme} with noise width $a=1/3$ and $S$ uniformly distributed in $[0,B]$ with  $B=a/2$.}}
\vspace{-13pt}
\label{fig:plot}
\end{figure}
Recall that for $E_1$, we first apply stimulus $1$ until either the output falls in either $(1+a/2,1+a/2+B/2)$ or it is less than $1-a/2+B/2$. For $b \in [0,B/2)$,  the probability of  the first event occurring in any write attempt is $b/a$, and that of the second event occurring is $(B/2-b)/a$. Hence the probability of the first step being completed in each write attempt, is
\be
\frac{b}{a} + \frac{B/2-b}{a} = \frac{B}{2a}.
\ee
Therefore the average number of writes for the first step of $E_1$ is $2a/B$ for all $b \in [0,B/2)$. The probability of the first step ending by obtaining an output less than $1-a/2+B/2$ is
\[ \frac{(B/2 -b)/{a}} {{b}/{a} + (B/2-b)/{a}} = \frac{B/2 -b}{B/2}.  \]
When this event occurs, the average number of additional writes required (by applying stimulus $0$) is $a/(B/2-b)$. Thus for $b \in [0,B/2)$, the average number of writes for writing on $E_1$ is
\be
\expec[\# \text{ writes} \mid  E_1, S=b] =\frac{2a}{B} + \frac{B/2 -b}{B/2} \cdot \frac{a}{B/2-b} = \frac{4a}{B}.
\ee
Substituting in \eqref{eq:rew_ext_basic}, we obtain
\be
\expec[\# \text{ writes} \mid \text{exterior}] = \frac{4a}{B}\frac{1}{2} + \frac{2a}{B}\frac{1}{2} = \frac{3a}{B}.
\ee
Using this in \eqref{eq:rew_total}, we get
\be
\kappa = p \frac{a}{D} + (1-p) \frac{3a}{B}
\ee
which corresponds to \eqref{eq:pd_ext_cond_imp} with $m=1$ and $\delta_1=0$. We now modify the scheme slightly to reduce the average number of rewrites to the level stated in Theorem \ref{thm:ext_scheme}:
\be
\kappa = p \frac{a}{D} + (1-p) \left(\frac{3a}{B} + \ln\frac{1}{1-\delta_1} +  \frac{\delta_1}{1-\delta_1} \ln \delta_1 \right)
\ee
with $\delta_1$ given by Table \ref{tab:deli}.

\emph{Optimizing the Switching Strategy}: To write on exterior region $1$ in Figure \ref{fig:two_regions}, the above coding scheme switches from stimulus $1$ to stimulus $0$ when  an output less than $1 - \tfrac{a}{2} + \tfrac{B}{2}$ is obtained. Such an output indicates that the value of the hidden state $S$ is less than $\tfrac{B}{2}$ which implies that the right bin of $E_1$  --  the region $[1+ \tfrac{a}{2} , 1+ \tfrac{a}{2} + \tfrac{B}{2}]$ --  is not \emph{fully} accessible with stimulus $1$; so the schemes switches to targeting the left bin of $E_1$ with stimulus  $0$. This switching strategy is not optimal. Consider a more general switching strategy of the following form: switch from stimulus $1$ to $0$ once you obtain an output less than $1 -\tfrac{a}{2}+ \tfrac{B}{2}(1-\delta_1)$ for some
$\delta_1 \in [0,1)$. This corresponds to switching once you detect that $S$ is less than $\tfrac{B (1-\delta_1)}{2}$. We now determine the optimum value of $\delta_1$ that minimizes the average number of rewrites. By symmetry, the switching strategy for exterior regions $E_2$ is to switch from stimulus $0$ to $1$ when you get an output greater than
$\frac{a}{2}+ \frac{B}{2} (1+\delta_1)$.

\begin{figure}[t]
\centering
\includegraphics[width=3in]{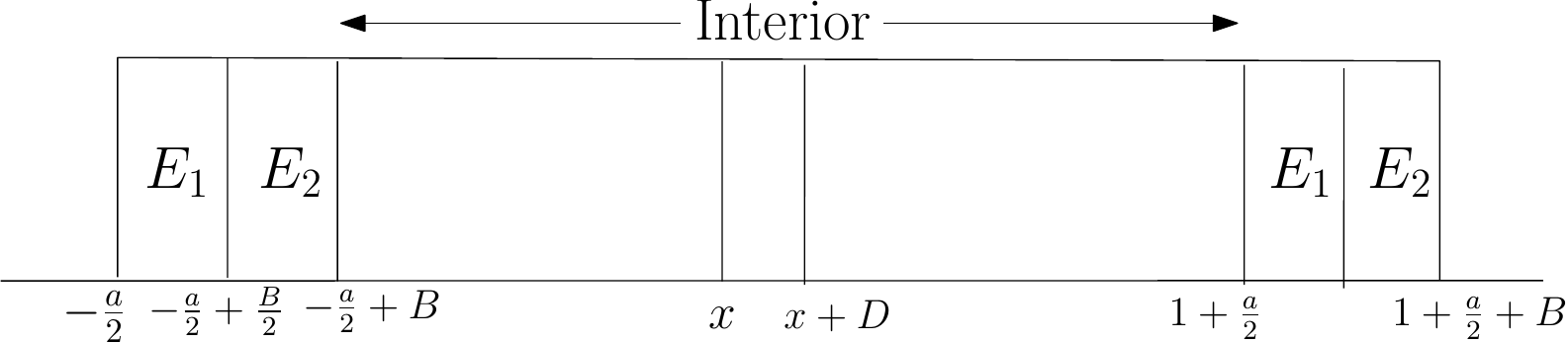}
\caption{\small{Construction $2$ with two exterior target regions.}}
\label{fig:two_regions}
\vspace{-11pt}
\end{figure}

The average number of rewrites for region $E_1$ is
\be
\begin{split}
\expec[\# \text{ writes} \mid \text{region } E_1]
= \int_{0}^{B} \expec[\# \text{ writes} \mid  E_1, S=b] \frac{1}{B} db
\end{split}
\label{eq:rew_ext_imp}
\ee
where $\expec[\# \text{ writes} \mid  E_1, S=b]$ with the new switching strategy can be calculated to be
\be
\begin{array}{ll}
{2ma}/{B}  &   \text{for }  \frac{B}{2} \leq b \leq B,\\
& \\
{a}/{b}  &   \text{for } \frac{B}{2}(1- \delta_1)   \leq b <  \frac{B}{2},\\
& \\
\frac{2a}{B(1-\delta_1)} \left( 1 +  \frac{(1- \delta_1)B - 2b}{B - 2b} \right) & \text{for }  0 \leq b <   \frac{B}{2}(1-\delta_1).
\end{array}
\ee
Substituting this in \eqref{eq:rew_ext_imp} and calculating the integral, we obtain
\be
\expec[\# \text{ writes} \mid  E_1]=\frac{a}{B}\left(3 + \ln\frac{1}{1-\delta_1}  + \frac{\delta_1}{1-\delta_1} \ln {\delta_1} \right).
\label{eq:del1_expec}
\ee
Using \eqref{eq:del1_expec} in \eqref{eq:rew_total} completes the proof for $m=1$.

\subsection{$2m$ Exterior Target Regions}
\emph{Coding Scheme}: Refer Figure \ref{fig:scheme2}. Writing on the interior regions is the same as before: For interior region $[x, x+D]$, repeatedly apply stimulus $(x+D-a/2)^+$. To write on exterior region $i, \ 1\leq i \leq 2m$:
\begin{itemize}
\item If $1\leq i \leq m$, write $1$ until the output falls in region $E_i$ or it is less than $1-\tfrac{a}{2} + \tfrac{B(i-\delta_i)}{2m}$,
In the first case, stop. In the second case, switch to writing $0$ until the output falls in the left bin of region $i$.

\item If $m+1 \leq i \leq 2m$, write $0$ until the output falls in region $E_i$ or it is greater than $\tfrac{a}{2} + \tfrac{B(i-1+ \delta_{2m+1-i})}{2m}$.
In the first case, stop. In the second case, switch to writing $1$ until the output falls in the right bin of region $i$.
\end{itemize}

\emph{Analysis}:
The rate calculation is straightforward, the only change from the previous subsection being that each of the exterior target regions  now represents $\log 2m$ bits of information. The average number of writes for an interior target region is $a/D$. For an exterior region, we calculate it separately for  each $E_i, i=1,\dots, m$ as follows. Note that by symmetry, $E_i$ and $E_{2m+1-i}$ have the same average cost. We have
\be
\begin{split}
\expec[\# \text{ writes} \mid \text{region } E_i]
= \int_{0}^{B} \expec[\# \text{ writes} \mid  E_i, S=b] \frac{1}{B} db
\end{split}
\label{eq:rew_ext_impi}
\ee
where $\expec[\# \text{ writes} \mid  E_i, S=b]$ can be calculated to be
\be
\begin{array}{ll}
{2ma}/{B}  &   \text{\small{for} }  \frac{Bi}{2m} \leq b \leq B,\\
& \\
\frac{2ma}{2mb- (i-1)B}  &   \text{\small{for} } \frac{B(i- \delta_i)}{2m}   \leq b <  \frac{Bi}{2m},\\
& \\
\frac{2ma}{B(1-\delta_i)} \left( 1 +  \frac{(i- \delta_i)B - 2mb}{iB - 2mb} \right) & \text{\small{for} }  \frac{B(i-1)}{2m} \leq b <   \frac{B(i-\delta_i)}{2m}, \\
& \\
\frac{2ma}{(i-\delta_i)B-2mb} + \frac{2ma}{B} &  \text{\small{for} }  0\leq b <   \frac{B(i-1)}{2m}.
\end{array}
\ee
Using this in \eqref{eq:rew_ext_imp} and calculating the integral, we obtain that for $1 \leq i \leq m$:
{\small{
\be
\expec[\# \text{ writes} \mid  E_i]=\frac{a}{B}\left(2m+1 + \ln\frac{i-\delta_i}{(1-\delta_i)^2}  - \frac{\delta_i}{1-\delta_i} \ln \frac{1}{\delta_i} \right).
\label{eq:deli_expec}
\ee
}}
The average number of write attempts for an exterior region is therefore
\ben
\frac{a}{B}\left[2m+1 + \frac{1}{m}\sum_{i=1}^{m} \left(\ln\frac{i-\delta_i}{(1-\delta_i)^2}  - \frac{\delta_i}{1-\delta_i} \ln \frac{1}{\delta_i} \right) \right].
\een
For each $i \in \{1, \ldots, m\}$, it is easily verified that the $\delta_i \in [0,1)$ that minimizes \eqref{eq:deli_expec} satisfies \eqref{eq:opt_deli_thm}.
This completes the proof.

\section{Conclusion} \label{sec:conc}
In a channel with unknown parameters (modeled by a hidden state), rewrites increase the capacity in two ways: 1) by mitigating the effect of write noise, and 2) by enabling the write controller to get progressively better estimates of the state.  For the uniform noise channel,  one of the key observations was that the hidden state does not affect coding in the interior region. This idea could be generalized to other channels where the output has bounded support.

There are many open questions to be explored. One is obtaining a capacity upper bound, which is challenging as we need to consider all adaptive input strategies.  The general capacity lower bound can be improved via a scheme that does simultaneous coding and estimation; the challenge here lies in analyzing such a coding scheme to get a computable expression for the achieved rate.  Another goal  is to modify the superposition scheme so that it robust to small amounts of read noise. As discussed at the end of Section \ref{sec:gen_lb},  the current scheme requires the reads to be highly accurate.

The channel model analyzed here is motivated by non-volatile memories such as Phase Change Memory and Resistive RAM. The coding schemes illustrate how information-theoretic techniques like superposition can be used to increase the storage density. Though the schemes presented are for analog storage channels, the  ideas can be extended to finite alphabet channels which arise in technologies such as Magnetic RAM \cite{mram05}.  A more sophisticated channel model for real memories is  one where the value written on the cell depends on the stimulus as well as the previous value stored in the cell.  Another interesting possibility is extending the model to account for stochastic variation of the cell contents over time, a phenomenon which is encountered in most memory technologies and which manifests as a read noise at the ``receiver''. We believe that developing efficient rewritable schemes for such realistic models will have a significant impact on several non-volatile memory technologies.

%

\appendix

\emph{Proof of Theorem \ref{thm:cont_ch_lb}}:

Fix an estimation period  $l \in \{0,\ldots, \lfloor\kappa -1 \rfloor\}$, an estimator $\hat{S}(l)$, a distribution $P_{U|\hat{S}(l)}$, and a function $f$ to generate the channel input $X=f(U, \hat{S}(l))$. This defines a joint distribution of $(S, \hat{S}(l), U, X, Y)$  in the set $\mc{P}$.

Construct a codebook consisting of $2^{nR_1}$ codewords, whose elements are picked i.i.d. according to $P_U$, the marginal distribution of the auxiliary random variable $U$. This codebook is partitioned in $2^{nR}$ bins where
\be R < I(U;Y) - I(U; \hat{S}(l)). \label{eq:R_final_cond} \ee
$R_1>R$ will be specified later.

The output space of each cell is divided into $\lfloor \kappa -l \rfloor$ target regions, as described in Section \ref{subset:awgn}  (see Figure \ref{fig:superpos}).

\emph{Encoding}: The message to be stored in the $n$-cell array consists of two parts $(m_1, m_2)$, where $m_1 \in \{1, \ldots, 2^{nR} \}$ and
$m_2 \in \{1, \ldots, (\lfloor \kappa -l \rfloor)^n \}$. Let $\hat{\mathbf{S}}(l)$ be the state estimate obtained using the first $l$ writes. To encode the first part of the message, we choose a codeword $\mathbf{U}$ from the $m_1$th bin such that the pair $(\mathbf{U}, \hat{\mathbf{S}}(l))$ is jointly typical \cite[Section 8.2]{CoverT} according to the distribution described by the following joint density:
\be
P_{U, \hat{S}(l)} (u, \hat{s}) = \int_{\mc{S}} P_S(s) P_{\hat{S}(l)|S}(\hat{s} | s)  P_{U |\hat{S}(l)} (u | \hat{s}) ds.
\label{eq:uy_joint}
\ee
From rate-distortion theory \cite{CoverT}, such a codeword $\mbf{U}$ can be found with high probability if
\be
R_1 - R > I(U; \hat{S}(l)).
\label{eq:quant_cond}
\ee
\eqref{eq:quant_cond} gives a lower bound on the minimum number of codewords in each bin ($2^{n(R_1-R)}$) required for successful encoding.
 The input stimulus $\mathbf{X}= \{ X_i\}_{i=1}^n$ is generated symbol by symbol as  $X_i = f(U_i, \hat{S}_i(l))$.

 The second part of the message is conveyed through superposition coding. For cell $i$, apply stimulus $X_i$ until the output falls in the appropriate target region.
For any realization of the input stimulus and state, the output is equally likely to fall in each of the target regions with probability $\tfrac{1}{ \lfloor \kappa -l \rfloor}$. Hence the average number of writes required after the estimation period is  $\lfloor \kappa -l \rfloor$, and the average total writes per cell is $l + \lfloor \kappa -l \rfloor$.

\emph{Decoding}: The decoder attempts to find a codeword $\hat{\mbf{U}}$ that is jointly typical with the stored sequence $\mbf{Y}$ according to \eqref{eq:uy_joint}. If there is a unique such codeword, its bin is decoded as the message $m_1$. The target region containing the output of each cell gives the message $m_2$.
The codeword $\mbf{U}$ can be successfully decoded if the rate of the codebook satisfies
\be
R_1 < I(U;Y).
\label{eq:ch_cond}
\ee
Combining \eqref{eq:quant_cond} and \eqref{eq:ch_cond}, we conclude that $\mbf{U}$ can be successfully encoded and decoded if \eqref{eq:R_final_cond}
 is satisfied.

We have thus shown that a total of $R + \log \lfloor \kappa - l \rfloor$ bits/cell can be reliably stored and decoded with average write cost $l + \lfloor \kappa - l \rfloor $   as long as $R$ satisfies \eqref{eq:R_final_cond}.

\IEEEtriggeratref{11}

\bibliographystyle{ieeetr}
\bibliography{rew_channels}
\end{document}